\documentclass[a4paper]{article}
\usepackage[utf8]{inputenc}
\usepackage{geometry}

\usepackage{amsmath}
\usepackage{amsfonts,amssymb,amsthm}

\usepackage{xcolor}
\definecolor{linkColor}{RGB}{0, 128, 128}
\definecolor{citeColor}{RGB}{0, 112, 64}
\definecolor{urlColor}{RGB}{120, 0, 120}
\usepackage{hyperref}
\hypersetup{
    colorlinks=true,
    linkcolor=linkColor,
    citecolor=citeColor,      
    urlcolor=urlColor,
}
\usepackage{url}

\usepackage{mathrsfs} 

\theoremstyle{plain}
\newtheorem{thm}{Theorem}
\newtheorem{lem}[thm]{Lemma}

\newtheorem{clm}[thm]{Claim}

\theoremstyle{definition}
\newtheorem{defn}{Definition}

\usepackage{tikz}
\usetikzlibrary{quantikz} 

\newcommand{\ccO}{\mathscr{O}}
\newcommand{\ccP}{\mathscr{P}}
\newcommand{\ccQ}{\mathscr{Q}}

\newcommand{\cA}{\mathcal{A}}
\renewcommand{\>}{\rangle}
\newcommand{\<}{\langle}

\DeclareMathOperator{\tr}{\mathrm{Tr}}
\DeclareMathOperator{\Tr}{\mathrm{Tr}}

\newcommand{\OO}{\mathrm{O}}

\newcommand{\qAnd}{\quad\text{and}\quad}
\newcommand{\qqAnd}{\qquad\text{and}\qquad}

\newcommand{\wo}[1]{z_{#1}} 

\title{Hybrid Quantum-Classical Search Algorithms}
\author{Ansis Rosmanis\thanks{E-mail: \texttt{rosmanis@math.nagoya-u.ac.jp}}\\ [.5ex]
\normalsize  Graduate School of Mathematics \\ 
\normalsize Nagoya University}
\date{June 6, 2023}

\begin{document}

\maketitle

\begin{abstract}
Search is one of the most commonly used primitives in quantum algorithm design. It is known that quadratic speedups provided by Grover's algorithm are optimal, and no faster quantum algorithms for Search exist. While it is known that at least some quantum computation is required to achieve these speedups, the existing bounds do not rule out the possibility of an equally fast hybrid quantum-classical algorithm where most of the computation is classical. In this work, we study such hybrid algorithms and we show that classical computation, unless it by itself can solve the Search problem, cannot assist quantum computation. In addition, we generalize this result to algorithms with subconstant success probabilities.
\end{abstract}

\section{Introduction}

In recent years, quantum computing has come closer to practicality. However, the speed of its operations and precision is still nowhere near to what the classical computers can do, so it is important, especially in the near future, to take as much advantage of classical computation as possible. Therefore it is essential to understand what parts of hybrid quantum-classical computation can be delegated to classical processing, and what parts are essentially quantum. Especially for tasks that are omnipresent in quantum computation.

The Search is one of the most fundamental and most studied problems in the quantum computing, with the famous Grover's algorithm achieving quadratic speedup over purely classical computation~\cite{grover:search}. However, it is still not fully understood under what conditions this quadratic speedup can be preserved.\footnote{For example, we know that it is not preserved under parallelization~\cite{zalka:Grover} and under certain noise models~\cite{regev:faultySearch}.} In this work, we investigate if one can save on the amount of quantum computation necessary for Search by employing classical computation. We show that the answer to this question is essentially ``no'': asymptotically, no amount of classical computation can reduce the number of quantum operations needed, unless this classical computation on its own can solve the Search problem.

More precisely, we consider search over a domain of size $n$, where our goal is to find a marked element $i$. Whether an element is marked or not is given by a black-box function $x$, which is the input of the problem. We consider algorithms that are given both an access to a quantum oracle for $x$ that can evaluate $x$ in a superposition as well as a classical access to $x$, which delegates the computation of $x$ to a classical oracle. While it is well known that any query algorithm for Search has to evaluate $x$ at least $\Omega(n^{1/2})$ times, one may ask: Can most of those evaluations be classical? For example, could it be possible to construct a hybrid quantum-classical search algorithm that uses $n^{1/3}$ quantum queries and $n^{1/2}$ classical queries to $x$? We answer this question in negative: any hybrid quantum-classical algorithm for Search needs to perform at least $\Omega(n^{1/2})$ quantum queries or at least $\Omega(n)$ classical queries.

\begin{thm}
\label{thm:mainTechA}
Suppose the algorithm runs on an input $x$ of Hamming weight $1$, makes $\tau_c$ classical queries and $\tau_q$ quantum queries, and finds the unique marked element with probability at least $1-\epsilon$. Then $\tau_c+4\sqrt{n}\tau_q \ge n(1-2\sqrt{\epsilon}-4/n^{1/4})$.
\end{thm}

While this theorem tells us, as for our example, that a hybrid algorithm that uses $\tau_q=n^{1/3}$ quantum queries and $\tau_c=n^{1/2}$ classical queries cannot solve the Search problem, it does not tell us what are the best possible chances of success of such an algorithm. To be more precise, because the success probability of an algorithm can be boosted by running it multiple times, the theorem above tells us that the success probability of the algorithm is at most $\mathrm{O}(\tau_q/n^{1/2}+\tau_c/n)$, which for our example would be $\mathrm{O}(n^{-1/6})$. If there were only quantum queries available, a better bound on the success probability could be proven using the amplitude amplification, but the presence of classical queries thwarts such an approach. Nevertheless, using a direct approach, we still succeed at proving the following tight bound, which states that the success probability of a hybrid $\tau_q$ quantum query and $\tau_c$ classical query algorithm for Search is at most $\mathrm{O}(\tau_q^2/n+\tau_c/n)$.

\begin{thm}
\label{thm:mainTechB}
Suppose the algorithm runs on an input $x$ of Hamming weight $1$ and makes $\tau_c$ classical queries and $\tau_q$ quantum queries. The probability that the algorithm finds the unique marked element is at most $(2\sqrt{\tau_c}+2\tau_q+1)^2/n$.
\end{thm}

Thus, even in scenarios where one is concerned with small success probabilities, as, for example, for cryptographic applications, a hybrid algorithm cannot substantially benefit from classical queries.


\subsection{Techniques}

We could broadly divide quantum query lower bound techniques into two groups.
One group is that of polynomial methods~\cite{beals:pol,aaronson:laurent}, which utilize the fact that the output probability of the algorithm can be represented by a polynomial in input values. The maximum possible degree of this polynomial is related to the number of queries performed, and the lower bounds are obtained by showing that polynomials approximating the value of the problem must have large degree.
The other group is progress-based methods, which define a certain progress function that is based on comparing the execution of the algorithm on various inputs.%
\footnote{The input might also be switched in the midst of computation, as in the hybrid method~\cite{bennett:searchLowerB}.}
The lower bounds are proven by showing both that, in order for the algorithm to achieve the desired success guarantees, the final value of the progress function must differ significantly from its initial value and that a single oracle call cannot change this value by much. This latter group contains the so called ``hybrid method''~\cite{bennett:searchLowerB}, which first showed the asymptotic optimality of Grover's algorithm, various methods for proving exact optimality of Grover's algorithm~\cite{zalka:Grover,dohotaru:GroverOpt}, the adversary method~\cite{ambainis:adversary} and many of its generalizations, and other methods. The lower bounds in this paper are also based on certain progress measures, some of which have been already considered before, some of which, as far as we know, are new.

Most progress-based methods use the fact that, for quantum query algorithms, without loss of generality the quantum memory remains in a pure state throughout the execution of the algorithm. However, for the hybrid algorithms considered here, this is not necessarily the case due to the classical queries. To overcome this complication, we slightly strengthen the classical oracle by giving it some quantum features of the standard quantum oracle, but still not making it too strong to prove the desired bounds. We call it the \emph{pseudo-classical} oracle, and it has a very useful property that it preserves purity when the input is the all-zeros string $0^n$.%
\footnote{Here, as is customary, we consider the input function $x\colon\{1,\ldots,n\}\rightarrow\{0,1\}\colon i\mapsto x_i$ to be given as an $n$-bit string $x_1\ldots x_n$. An index $i$ is marked if $x_i=1$.}
 In addition, when the algorithm is run on inputs with Hamming weight $1$, that is, inputs with a unique marked element, each oracle call introduces only a small amount of entropy on average. Namely, while the memory might be in a mixed state, when interpreting this mixed state as a convex combination of pure density matrices, one pure term clearly dominates.

To prove the bound for the large (i.e., constant) success probability case, Theorem~\ref{thm:mainTechA}, we use a progress measure that is solely based on the Euclidean distances between the pure state corresponding to input $0^n$ and the ``dominant pure component'' corresponding to inputs of Hamming weight $1$. However, this progress measure fails when we want to consider small (i.e., subconstant) success probabilities. In that case, we proceed by introducing two separate progress measures, and track how queries to quantum and pseudo-classical oracles may affect them. For pseudo-classical queries, there is a tradeoff between possible changes in those two progress measures, and dealing with this tradeoff is the technically most involved part of this work.

%

\section{Connections to other works}

The pseudo-classical oracle, independently from the current work and under the name of semi-classical oracle, was introduced by Ambainis, Hamburg, and Unruh~\cite{ambainis:semiclassical}. They used this oracle to improve upon certain security bounds in post-quantum cryptography. In regards to the Search problem, similarly to the present work, they showed that a quantum algorithm that queries the pseudo-classical oracle $\tau_p$ times cannot find a marked element with probability more than $4\tau_p/n$. They did not, however, consider search with access to both the quantum and pseudo-classical oracle.

Our progress function for the large success probability case---the Euclidean distance between two vectors, not necessarily of unit norm, squared---is used by Regev and Schiff~\cite{regev:faultySearch} when addressing Search with faulty oracle. 
Boyer, Brassard, H{\o}yer, and Tapp used the same quantity, except for unit vectors, to show an almost tight (up to less than three percent) lower bound for Search~\cite{boyer:groverTight}.

Before introducing tailor-made methods for the problem at hand, it is natural to ask if we can use some of the existing general lower bound frameworks, such as the polynomial method or the adversary method, to prove the desired bounds. It is not clear if these methods apply, however, at least some specific cases can be handled by them. In particular, consider the scenario where all $\tau_c$ classical queries are performed before any quantum queries. In this case, it is easy to see that the best thing that can be done by the classical queries is to uniformly at random choose $\tau_c$ indices and query the input $x$ at them. This random procedure can be encoded in a pure, $x$-dependent state that, when measured in the standard basis, would provide these $\tau_c$ random index-value pairs. With such a state, the tight lower bound (for the large success probability case) can be obtained by the standard state-conversion adversary bound~\cite{lee:stateConversion}. It is reasonable to believe that a lower bound based on Laurent polynomials~\cite{aaronson:laurent} could also handle this special case.

Query lower bounds for algorithms with access to multiple oracles was also studied in~\cite{kimmel:OraclesWCosts, belovs:approxCount}. However, those works only considered quantum oracles. In particular, for every oracle, one also had access to its inverse.

In the near future, quantum computers are expected to outperform their classical counterparts for only a limited number of tasks, therefore considering hybrid computation is a natural research direction.
While this work focusses on quantum algorithms that have access to both the quantum and the classical oracle, various other hybrid quantum-classical models of computation have been considered before (see, for example,~\cite{FarhiGG2014:QAOA,ChiaCL:NeedDepth:2023}). In fact, even most fault-tolerant quantum computation schemes involve classical components and can thus be thought of as hybrid quantum-classical computation.

\paragraph{Recent developments.}

While in this work we have shown that there are no non-trivial tradeoffs between quantum and classical query complexity for the Search problem, there are fundamental problems for which such tradeoffs exist. Consider the Collision problem, which asks to distinguish if a function is one-to-one or two-to-one given that either is the case. Due to the birthday bound, the classical query complexity of this problem $\Theta(n^{1/2})$, while the quantum query complexity is $\Theta(n^{1/3})$ \cite{brassard:collision,aaronson:collisionLowerOriginal,shi:collisionLowerOriginal}. However, there exists a hybrid quantum-classical algorithm that uses $\tau_c$ classical queries and $\tau_q=O(\sqrt{n/\tau_c})$ quantum queries. 
An earlier version of this work asked the following question: Can we show that this tradeoff $\tau_q\sqrt{\tau_c}=O(\sqrt{n})$ is optimal when $\tau_q < \tau_c <  o(\sqrt{n})$?

Recently,  Hamoudi, Liu, and Sinha \cite{hamoudi:2022:tradeoffs} studied a very similar problem, essentially answering that question in affirmative. More precisely, they considered a version of the Collision problem where one is given the quantum and the classical oracle access to the same uniformly random function $f\colon\{1,\ldots,m\}\rightarrow\{1,\ldots,n\}$ and one has to find a collision in this function. They showed that a quantum query algorithm that makes $\tau_c$ classical queries and $\tau_q$ quantum queries has the success probability of finding the collision no more than $\OO((\tau_c^2+\tau_c\tau_q^2+\tau_q^3)/n)$, which is tight~\cite{brassard:collision}. They also showed that, in the same random function setting, the success probability of Search is at most $\OO((\tau_c+\tau_q^2)/n)$, essentially providing an alternative proof of Theorem~\ref{thm:mainTechB}.
Their proof techniques are different from the ones considered here and are based on 
 Zhandry's compressed oracle framework~\cite{zhandry:record}.

\section{Model of Computation}

We assume that the reader is familiar with basic concepts of quantum computation. For an introductory text, see, for example~\cite{NielsenChuang}.

\subsection{Quantum Memory}

The memory of a quantum algorithm is organized in registers. Each register is associated with some finite set $S$ and a complex Euclidean space of dimension $|S|$, denoted $\mathbb{C}^{S}$. The standard basis of this space is some fixed orthonormal basis whose vectors are uniquely labeled by the elements of $S$. The pure states of the register are denoted by (column) unit vectors in $\mathbb{C}^{S}$. A qubit is a register associated with the set $\{0,1\}$, and multiple qubits can be grouped together into a larger register.

For $M$ a matrix or a column vector, let $M^*$ denote its complex conjugate transpose.
We use the bold font to denote standard basis vectors. That is, given $i\in S$,  the vector $\pmb{i}$ corresponds to $|i\>$ in Dirac's notation, and thus $\pmb{i}\pmb{i}^*$ is equivalent to $|i\>\<i|$.  We use regular (i.e., not bold) letters $\psi$ and $\zeta$ to denote other vectors, which are not required to be one of the standard basis vectors, nor to have the unit length. We may also write a pure state $\psi$ as its corresponding density operator $\psi\psi^*$. The quantum memory can also be in a mixed-state, in which case its corresponding density operator is a convex combination over pure states $\psi\psi^*$.

The memory of a quantum algorithm typically consists of multiple registers, and the state of the entire memory is a density operator on the tensor product of the Euclidean spaces corresponding to each register. The evolution of quantum memory is governed by completely-positive trace-preserving super-operators, an evolution governed by unitary operators being a special case. When an operator or a super-operator acts as the identity on some registers, we typically omit those registers from the notation.

For density operators $\rho$ and $\sigma$ acting on the same space, the trace distance between them is $D(\rho,\sigma):=\Tr[|\rho-\sigma|]/2$ and the fidelity between them is $F(\rho,\sigma):=(\Tr[\sqrt{\rho^{1/2}\sigma\rho^{1/2}}])^2$. For a pure state $\psi\psi^*$, we have $F(\rho,\psi\psi^*)=\psi^*\rho\psi$. 

\subsection{Oracles}
\label{ssec:oracles}

We consider algorithms whose inputs are $n$-bit strings; equivalently, we may think of them as functions from $\{1,\ldots,n\}$ to $\{0,1\}$. For an input $x\in\{0,1\}^n$, we call an index $i\in\{1,\ldots,n\}$ \emph{marked} if $x_i=1$.
A query algorithm accesses the input through calls to various oracles. We use terms ``oracle calls'' and ``queries'' interchangeably. 
We consider three types of oracles: the classical oracle $\ccO$ and the quantum oracle $\ccQ$, which are originally used by the algorithm, as well as the pseudo-classical oracle $\ccP$, which we introduce for proof purposes.

All three oracles involve two registers, the \emph{index} register corresponding to the set $\{1,\ldots,n\}$ and a qubit called a \emph{target} register. While a call to quantum oracle is a unitary operation acting on both registers, calls to classical and pseudo-classical oracles are non-unitary operations that start without a target register and introduce the target register as they are executed. Thus, as the algorithm proceeds, its memory use will increase.

In the following definition, $\rho$ is a linear operator on the index register, $\sigma$ is a linear operator on the joint index and target register, $i\in\{1,\ldots,n\}$, and $b\in\{0,1\}$. We define oracles $\ccO,\ccP,\ccQ$ as super-operators expressed via Kraus operators. 

\begin{defn} (Oracles $\ccO,\ccP,\ccQ$)
The classical oracle corresponding to input $x$ is defined as
\[
\ccO_x\colon \rho\mapsto \sum\nolimits_{i}O_{x,i}\rho O_{x,i}^*
\qquad\text{where}\qquad
O_{x,i}:= \pmb{i}\pmb{i}^* \otimes \pmb{x_i}.
\]
The pseudo-classical oracle corresponding to input $x$ is defined as
\[
\ccP_x\colon \rho\mapsto \sum\nolimits_{b}P_{x,b}\rho P_{x,b}^*
\qquad\text{where}\qquad
P_{x,b}:= \sum_{i\colon x_i=b}\pmb{i}\pmb{i}^* \otimes \pmb{b}.
\]
The quantum oracle corresponding to input $x$ is defined as
\[
\ccQ_x\colon \sigma\mapsto Q_x\sigma Q_x
\qquad\text{where}\qquad
Q_x=\sum_{i,b} \pmb{i}\pmb{i}^*\otimes(\pmb{b\oplus x_i})\pmb{b}^*.
\]
\end{defn}

Note that $(\pmb{b\oplus x_i})$ is $\pmb{1}$ if $x_i\ne b$ and $\pmb{0}$ otherwise.  
$Q_x$ in a $2n\times 2n$ unitary, while Kraus operators $O_{x,i}$ and $P_{x,b}$ have dimensions $2n\times n$.

Operationally, the three oracles act as follows. The classical oracle measures the index register in the standard basis, and, if the measurement outcome is $i$, it appends a qubit in state $\pmb{x_i}$ to the already existing memory.
The pseudo-classical oracle, given the state of the index register is $\pmb{i}$, appends a qubit in state $\pmb{x_i}$ to the already existing memory, and then measures this appended qubit.
The quantum oracle applies the unitary $Q_x$ to the joint index-target register.
See Figure~\ref{fig:diags} for implementations of $\ccP_x$ and $\ccO_x$ through the quantum oracle $\ccQ_x$.
Also note that the classical oracle can be simulated by the pseudo-classical oracle by measuring the content of the index register in the standard basis.

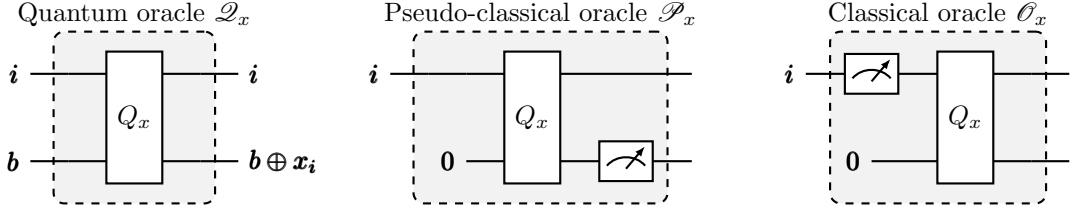
\begin{figure}[h!]
\centering
\begin{tikzcd}
\lstick{$\pmb{i}$} & \qw
\gategroup[2,steps=3,style={dashed, rounded corners, fill=gray!10, inner xsep=2pt}, background]{Quantum oracle $\ccQ_x$}
  &\gate[2]{Q_x} & \qw& \qw \rstick{$\pmb{i}$}& & & &
 \lstick{$\pmb{i}$} & \qw
\gategroup[2,steps=4,style={dashed, rounded corners, fill=gray!10, inner xsep=2pt}, background]{Pseudo-classical oracle $\ccP_x$}
 & \qw &\gate[2]{Q_x} & \qw& \qw & & &
\lstick{$\pmb{i}$} & \meter{}
\gategroup[2,steps=3,style={dashed, rounded corners, fill=gray!10, inner xsep=2pt}, background]{Classical oracle $\ccO_x$}
   &\gate[2]{Q_x} & \qw& \qw &
\\
\lstick{$\pmb{b}$} & \qw  &  & \qw & \qw \rstick{$\pmb{b\oplus x_i}$} & & & &
 & & \lstick{$\pmb{0}$}  &  & \meter{} & \qw & & &
& \lstick{$\pmb{0}$}  &  & \qw & \qw 
\end{tikzcd}
\caption{
Circuit diagrams for the three oracle types. Measurement operators denote non-destructive projective measurements on the corresponding standard bases. Quantum oracle $\ccQ_x$ corresponds to applying a unitary operator $Q_x$, and the other two oracles can be implemented via access to $Q_x$. Those oracles are the operations contained within the dashed lines.}
\label{fig:diags}
\end{figure}

We note that Grover's algorithm always calls the quantum oracle with the target register being in state $e_-:=(\pmb{0}-\pmb{1})/\sqrt{2}$, therefore effectively mapping $\pmb{i}$ of index register to $(-1)^{x_i}\pmb{i}$. We will also use the state $e_-$ when proving lower bounds.

For illustrative purposes, let us consider an example of the application of the oracles, where we treat mixed states as probabilistic mixtures of pure states.
Let the input be $x=0010\ldots0$, which has a unique marked index $3$.
For a unit vector $\alpha_3\pmb{3}+\alpha_5\pmb{5}+\alpha_6\pmb{6}$, the three oracles act as follows:
\begin{alignat*}{3}
& \ccO_x
&\;\colon\;&\hspace{3.6pt}
 \alpha_3\pmb{3}+\alpha_5\pmb{5}+\alpha_6\pmb{6}
 &\quad\mapsto\quad&
 \begin{cases}
 \pmb{3}\otimes\pmb{1}\hspace{45pt} & \text{w.p. }|\alpha_3|^2,\\
 \pmb{5}\otimes\pmb{0} & \text{w.p. }|\alpha_5|^2, \\
  \pmb{6}\otimes\pmb{0} & \text{w.p. }|\alpha_6|^2,
   \end{cases}
\\
& \ccP_x
&\;\colon\;&\hspace{3.6pt}
 \alpha_3\pmb{3}+\alpha_5\pmb{5}+\alpha_6\pmb{6}
 &\quad\mapsto\quad&
 \begin{cases}
    \pmb{3}\otimes\pmb{1}\hspace{45pt} & \text{w.p. }|\alpha_3|^2,\\
  \frac{\alpha_5\pmb{5}+\alpha_6\pmb{6}}{\sqrt{|\alpha_5|^2+|\alpha_6|^2}} \otimes \pmb{0} & \text{w.p. }|\alpha_5|^2+|\alpha_6|^2,
   \end{cases}
 \\
 & \ccQ_x
&\;\colon\;&
( \alpha_3\pmb{3}+\alpha_5\pmb{5}+\alpha_6\pmb{6})\otimes\pmb{0}
 &\quad\mapsto\quad&
 \alpha_3\,\pmb{3}\otimes\pmb{1}
 +
 (\alpha_5\pmb{5}+\alpha_6\pmb{6})\otimes\pmb{0}.
\end{alignat*}


\subsection{Hybrid Quantum-Classical Query Algorithms}

We divide the space of the algorithm into two registers: one is the index register and the other is the \emph{workspace} register that we assume to consist of some number of qubits. As we have described above, every call to the classical oracle $\ccO$ introduces an extra qubit, a target register, and we incorporate it into the workspace register. For every call to the quantum oracle $\ccQ$, we designate one already existing qubit of the workspace register as the target register. 

A hybrid quantum-classical query algorithm is specified by four components: (1) the number and order of classical and quantum queries, (2) the initial state of the algorithm, (3) input-independent unitary operators that govern the evolution of the quantum system between oracle calls, and (4) the final measurement. Let us describe these components in detail.
\begin{enumerate}
\item
We assume that the sequence in which the calls to two oracles are performed is fixed up front. In particular, it is independent from the input and we do not allow the algorithm to decide which oracle to call based on contents of some (classical) memory. 
 Let $\tau_c$ and $\tau_q$ denote, respectively, the total number of calls to the classical and the quantum oracle that the algorithm performs, and let $\tau:=\tau_c+\tau_q$. We enumerate oracle calls from $1$ to $\tau$. Let $T_c\subseteq\{1,\ldots,\tau\}$ be the set of all $t$ such that $t$-th oracle call is classical and let $T_q:=\{1,\ldots,\tau\}\setminus T_c$ be the set of all $t$ such that $t$-th oracle call is quantum; we have $|T_c|=\tau_c$ and $|T_q|=\tau_q$.
\item Let initially the workspace register consist of $\ell_0$ qubits, therefore the entire initial memory corresponds to an $n2^{\ell_0}$-dimensional Euclidean space. The initial state of the algorithm is an input-independent pure state $\psi^0$ in this space; the first oracle call is performed directly on this state.
\item
Among the first $t$ oracle calls, $|T_c\cap\{1,\ldots,t\}|$ are classical, each expanding the workspace register by a qubit. 
Let $\ell_t:=\ell_0+|T_c\cap\{1,\ldots,t\}|$. The evolution between oracle calls and after the last call is given by input-independent unitary operators $U_1,\ldots,U_\tau$, where the dimension of $U_t$ is $n2^{\ell_t}$.
\item
Given some finite set $\cA$ of answers, the final measurement is given by a set $\{\Pi_a\colon a\in\cA\}$, where each $\Pi_a$ is an orthogonal projector of dimension $n2^{\ell_0+\tau_c}$ and $\sum_a\Pi_a=I$.
\end{enumerate}
The execution of the algorithm starts in the initial state $\psi^0$, and then alternates between oracle calls and input-independent unitaries as follows. Iteratively, for $t=1,2,\ldots,\tau$, the algorithm first performs an oracle call to $\ccO_x$ if $t\in T_c$ or $\ccQ_x$ if $t\in T_q$, and then applies unitary $U_{t}$. Finally, the algorithm performs a measurement according to $\{\Pi_i\colon a\in\cA\}$, returning the measurement outcome $a$ as an answer. We say that the algorithm is \emph{successful} if $a$ is a correct answer for input $x$, and we say that it \emph{fails} otherwise.

In this paper, we use the term \emph{query lower bounds} when showing that a certain task is hard for an algorithm. Thus, we may refer to upper bounds on the success probability as query lower bounds.

\section{Lower bound framework}

For proving that the Search problem has at least a certain computational hardness, we now replace the classical oracle $\ccO$ by the pseudo-classical oracle $\ccP$. Thus, in fact, we will prove stronger versions of Theorems~\ref{thm:mainTechA} and~\ref{thm:mainTechB}, given below, which instead of the classical oracle addresses the pseudo-classical oracle. Since the pseudo-classical oracle is at least as powerful as the classical one, the strengthened Theorems~\ref{thm:mainTechA} and~\ref{thm:mainTechB} given below imply their original versions given in the introduction.

From now on, we use $\kappa$ to denote values in $\{0,1,\ldots,n\}$ and $k$ to denote values in $\{1,\ldots,n\}$.
Our proofs will only require considering inputs of Hamming weights $0$ and $1$. Let $x^{(0)}:=0^n$ and $x^{(k)}:=0^{k-1}10^{n-k}$.
Given $\kappa\in\{0,1,\ldots,n\}$, let $\ccP_\kappa$, $P_{\kappa,0}$, $\ccQ_\kappa$, $Q_{\kappa}$ be short for, respectively, $\ccP_{x^{(\kappa)}}$, $P_{x^{(\kappa)},0}$, $\ccQ_{x^{(\kappa)}}$, $Q_{x^{(\kappa)}}$.
We prove the following theorems in Sections~\ref{sec:proofH} and \ref{sec:proofAB}.

\bigskip
\noindent
{\bf{}Theorem~\ref{thm:mainTechA}} \;(Strengthened){\bf.} \emph{Suppose the algorithm runs on an input $x$ of Hamming weight $1$, makes $\tau_c$ pseudo-classical queries and $\tau_q$ quantum queries, and finds the unique marked element with probability at least $1-\epsilon$. Then $\tau_c+4\sqrt{n}\tau_q \ge n(1-2\sqrt{\epsilon}-4/n^{1/4})$.}
\bigskip

\noindent
{\bf{}Theorem~\ref{thm:mainTechB}} \;(Strengthened){\bf.} \emph{Suppose the algorithm runs on an input $x$ of Hamming weight $1$ and makes $\tau_c$ pseudo-classical queries and $\tau_q$ quantum queries. The probability that the algorithm finds the unique marked element is at most $(2\sqrt{\tau_c}+2\tau_q+1)^2/n$.}
\bigskip

We remark the only difference in the statements of these theorems, compared to their respective original versions, is that the word ``classical'' has been replaced by ``pseudo-classical''. \bigskip

We note that, in the strengthened version of Theorem~\ref{thm:mainTechB} (the one concerning pseudo-classical queries, not classical queries), the constants $2$ in front of both $\sqrt{\tau_c}$ and $\tau_q$ cannot be improved.
As shown in \cite{ambainis:semiclassical}, for Search with unique marked element and a small number of queries $\tau$, the pseudo-classical oracle can outperform the classical oracle by about a factor of $4$.
Indeed, on the one hand, if an algorithm is given only $\tau$ classical queries (and no other types of queries), the best it can do is to query the oracle on randomly guessed indices and then, if the marked index is not found, guess it randomly, thus having the success probability
$(\tau+1)/n$.
On the other hand, if we run Grover's algorithm, but use the pseudo-classical oracle instead of the quantum oracle, the success probability becomes
$1-(1-1/n)^2(1-2/n)^{2(\tau_c-1)}$,
which is approximately $4\tau/n$ for small $\tau$. 

The authors of \cite{ambainis:semiclassical} also show that this factor $4$ improvement is essentially optimal, and Theorem~\ref{thm:mainTechB} can be thought of as a strengthening of such a statement. However, if we aim for a success probability close to $1$, Theorem~\ref{thm:mainTechA} states that the advantage of the pseudo-classical oracle over the classical oracle vanishes.

\paragraph{Dominant pure components.}

Even though the memory state of the algorithm for inputs of Hamming weight $1$ might become mixed, we will show is that, essentially, for an average input, one pure component of this state dominates unless the number of queries is large. That component will correspond the probabilistic ``branch'' of the computation where the pseudo-classical oracle has not found the solution, that is, the target register introduced by it has always been in state $\pmb{0}$. Let us now formally introduce these pure components.

Let $\kappa\in\{0,1,\ldots,n\}$ indicate that the algorithm is effectively run on input $x^{(\kappa)}$. For $t\in\{0,1,\ldots,\tau\}$, define vector $\psi_\kappa^t$ as follows. Let $\psi_\kappa^0:=\psi^0$ be the initial state of the algorithm, which is the same for all $\kappa$. Then, for $t\ge 1$, recursively define $\psi_\kappa^{t}=U_tQ_\kappa\psi_\kappa^{t-1}$ if $t$-th oracle call is quantum and $\psi_\kappa^{t}=U_tP_{\kappa,0}\psi_\kappa^{t-1}$ if $t$-th oracle call is pseudo-classical.

For $\kappa=0$, we have $P_{0,0}=I\otimes\pmb{0}$ and $P_{0,1}=0$. Therefore, when we run the algorithm on $x^{(0)}=0^n$, the final state is pure, in particular, it is $\psi_0^\tau$, with the corresponding density matrix $\psi_0^\tau(\psi_0^\tau)^*$. For other inputs, $x^{(k)}$ with $k\ge 1$, the final state might be mixed, but, due to linearity, it is at least $\psi_k^\tau(\psi_k^\tau)^*$ (in the semi-definite ordering). 
Hence, in order for the algorithm to distinguish between $x^{(0)}$ and $x^{(k)}$ with high probability, we need $\psi_k^\tau$ and $\psi_0^\tau$ to be far away. More precisely, we either need $\psi_k^\tau$ to have a small norm or to be almost orthogonal to $\psi_0^\tau$, or both.

All the progress measures that we consider in this paper will be determined by the inner product $\<\psi_0^t,\psi_k^t\>$ and length $\|\psi_k^t\|$.
Figure~\ref{fig:prog} illustrates these progress measures for a very specific case; more formally we will introduce them in Sections~\ref{sec:proofH} and \ref{sec:proofAB}.

\begin{figure}[h!]
  \centering
  \includegraphics[scale=0.5]{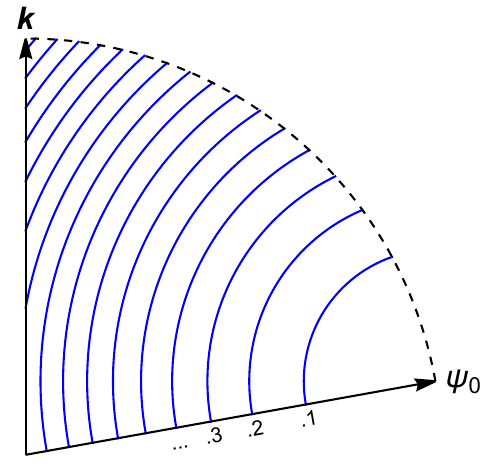}
\hspace{30pt}
\includegraphics[scale=0.5]{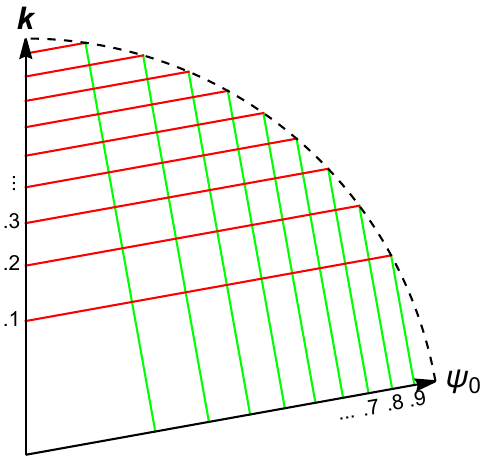}
\caption{
Progress measure $H_k$ on the left and progress measures $A_k$ (green, almost vertical lines) and $B_k$ (red, almost horizontal lines) on the right. Here, for simplicity, we assume that $\psi_k$ is a linear combination of $\psi_0=\sum_{k'=1}^{32}\pmb{k'}/\sqrt{32}$ and $\pmb{k}$ with non-negative coefficients. The progress measures are displayed in increments/decrements by $0.1$, starting from the initial case $\psi_k=\psi_0$, and the time $t$ superscripts are omitted.}
\label{fig:prog}
\end{figure}

\section{Case of Constant Success Probability}
\label{sec:proofH}

It is common to consider two versions of the Search problem: the \emph{detection} version, where one has to decide if any index is marked, and the \emph{finding} version, where one has to find a marked index assuming it exists. In this section we define a progress function $H$ and use it to prove the following lower bound for the detection version of Search.

 \begin{thm}
\label{thm:mainTechADec}
Suppose the algorithm makes $\tau_c$ pseudo-classical (or classical) queries and $\tau_q$ quantum queries, and distinguishes whether the Hamming weight of the input is $0$ or $1$ with probability at least $1-\epsilon$. Then $\tau_c+4\sqrt{n}\tau_q \ge n(1-4\sqrt{\epsilon})$.
\end{thm}

We note that the lower bound  $\tau_c+4\sqrt{n}\tau_q \ge n(1-4\sqrt{\epsilon})-1$ for the finding version can be obtained directly form Theorem~\ref{thm:mainTechADec}.
Indeed, if one had an algorithm for the finding version that on inputs of Hamming weight $1$ returned the correct answer $k\in\{1,\ldots,n\}$ with probability at least $1-\epsilon$, then one could turn such an algorithm into an algorithm that detects if the Hamming weight of input $x$ is $0$ or $1$ by simply performing one additional query (classical or quantum) that asks for the value of $x_k$, and then returning the Hamming weight of $x$ to be $x_k$. This new algorithm would be always correct on the input of Hamming weight $0$ and its success probability on the inputs of Hamming weight $1$ would be at least $1-\epsilon$. However, the lower bound obtained by this reduction has worse parameters than the bound given by Theorem~\ref{thm:mainTechA}.

The proofs of Theorems~\ref{thm:mainTechA} and \ref{thm:mainTechADec} are almost equivalent. We will conclude the proof of Theorem~\ref{thm:mainTechA} in Section~\ref{ssec:decProofEnd}, after we have introduced some machinery used to prove Theorem~\ref{thm:mainTechB}. This machinery, as a byproduct, provides a convenient way to relate the progress function $H$ to the success probability for the finding version of Search.

\subsection{Proof of Theorem~\ref{thm:mainTechADec}}

Define progress measures $H_k^{(t)}:=\|\psi_k^t-\psi_0^t\|^2=\|U_t^{*}\psi_k^t-U_t^{*}\psi_0^t\|^2$ and $H^{(t)}:=\sum_k H_k^{(t)}/n$. Initially, for $t=0$, we have $H^{(0)}=0$. The theorem is implied by the following two claims; Claim~\ref{clm:progQ} upper-bounds $H^{(\tau)}$ by $\tau_c/n+4\tau_q/\sqrt{n}$ while Claim~\ref{clm:progFinal} lower-bounds it by $1-4\sqrt{\epsilon}$. We note that the proofs of Claim~\ref{clm:progFinal} and the pseudo-classical oracle part of Claim~\ref{clm:progQ} closely follow \cite{regev:faultySearch}, where the authors analyzed Search with faulty oracle calls.


\begin{clm}
\label{clm:progFinal}
Suppose that for all $k\ge 0$ the algorithm returns the correct answer with probability at least $1-\epsilon$. Then $H^{(\tau)} \ge 1-4\sqrt{\epsilon}$.
\end{clm}

\begin{proof}
Let us write $\psi_\kappa:=\psi_\kappa^\tau$, $H_k:= H_k^{(\tau)}$, and $H:= H^{(\tau)}$ for short.
Let density operator $\rho_\kappa$ be the final state of the algorithm when it is run on input $x^{(\kappa)}$. Note that $\rho_0=\psi_0\psi_0^*$ and $\rho_k\succeq\psi_k\psi_k$.
Since the algorithm correctly distinguishes between inputs of Hamming weights $0$ and $1$ with the error probability at most $\epsilon$, the trace distance between $\rho_0$ and $\rho_k$ is at least $1-2\epsilon$. 
Thus, by the the Fuchs--van de Graaf inequalities,
\[
1-2\epsilon \le D(\rho_0,\rho_k) \le \sqrt{1-F(\rho_0,\rho_k)} = \sqrt{1-\psi_0^*\rho_k\psi_0} 
\le \sqrt{1-|\<\psi_0,\psi_k\>|^2}.
\]
Therefore $|\<\psi_0,\psi_k\>|^2 \le 4\epsilon-4\epsilon^2$ and $|\<\psi_0,\psi_k\>|\le 2\sqrt{\epsilon}$.
To conclude, we have
\[
H_k = \|\psi_k-\psi_0\|^2 = \|\psi_k\|^2+\|\psi_0\|^2 - 2\Re \<\psi_0,\psi_k\> \ge 1-2|\<\psi_0,\psi_k\>| \ge 1-4\sqrt{\epsilon},
\]
and thus $H \ge 1-4\sqrt{\epsilon}$. 
\end{proof}


\begin{clm}
\label{clm:progQ}
We have $H^{(t)}\le H^{(t-1)} +1/n$ when $t\in T_c$ and $H^{(t)}\le H^{(t-1)} +4/\sqrt{n}$ when $t\in T_q$.
\end{clm}

\begin{proof}
For sake of brevity, we omit the time superscripts and write $H$, $H_k$, $\psi_k$, $\psi_0$ for $H^{(t-1)}$, $H_k^{(t-1)}$, $\psi_k^{(t-1)}$, $\psi_0^{(t-1)}$, respectively, and $H^{+}$ and $H_k^{+}$ for $H^{(t)}$ and $H_k^{(t)}$, respectively.%
\footnote{\label{fn:superPlus}%
Here the superscript $+$ in the notation $H^{+}$ and $H_k^{+}$ is meant to indicate that we are considering the values of those progress functions one time step later than their corresponding values $H$ and $H_k$, respectively. We interpret the superscript $+$ the same way in the proof of Claim~\ref{clm:ABbound}.}
For $k\in\{1,\ldots,n\}$, define the projector $\Pi_k := \pmb{k}\pmb{k}^*$ if $t\in T_c$ and $\Pi_k:=\pmb{k}\pmb{k}^*\otimes e_-e_-^*$ if $t\in T_q$. 
 Define 
$\gamma_k:=\|\Pi_k\psi_0\|^2$, which satisfy $\sum_k\gamma_k\le 1$ (with equality when $t\in T_c$), and also define $\beta_k := \|\Pi_k\psi_k\|^2\le 1$.

\bigskip

\noindent
\emph{Pseudo-classical oracle.}
Consider $t\in T_c$, and note that $P_{k,0}=(I-\Pi_k)\otimes\pmb{0}$ and $ P_{0,0}=I\otimes\pmb{0}$. We have
\begin{align*}
H_k^{+}- H_k
& = \|P_{k,0}\psi_k-P_{0,0}\psi_0\|^2 - \|\psi_k-\psi_0\|^2
\\ & = \|(I-\Pi_k)\psi_k\otimes\pmb{0}-\psi_0\otimes\pmb{0}\|^2 - \|\psi_k-\psi_0\|^2
\\ & = \|\psi_k- \Pi_k\psi_k-\psi_0\|^2 - \|\psi_k-\psi_0\|^2
\\ & = \<\psi_0-\psi_k,\,\psi_0-\psi_k\> + \<\Pi_k\psi_k,\, \psi_0 - \psi_k\> + \<\psi_0 - \psi_k,\, \Pi_k\psi_k\> + \<\Pi_k\psi_k,\,\Pi_k\psi_k\>  - \|\psi_0-\psi_k\|^2
\\ & = \<\Pi_k\psi_k,\, \psi_0 - \psi_k\> + \<\psi_0 - \psi_k,\, \Pi_k\psi_k\> + \|\Pi_k\psi_k\|^2
\\ & = \<\Pi_k\psi_k,\, \Pi_k\psi_0\> - \<\Pi_k\psi_k,\, \Pi_k\psi_k\> + \<\Pi_k\psi_0,\, \Pi_k\psi_k\> - \<\Pi_k\psi_k,\, \Pi_k\psi_k\> + \|\Pi_k\psi_k\|^2
\\ & = 2\Re\<\Pi_k\psi_k,\Pi_k\psi_0\>  - \|\Pi_k\psi_k\|^2
\\ & \le \sqrt{\beta_k} \big(2\sqrt{\gamma_k}  - \sqrt{\beta_k})
\le \gamma_k,
\end{align*}
where we have used $\Pi_k=\Pi_k^2$ for the penultimate equality.
Hence $H^{+}- H\le \sum_k\gamma_k/n = 1/n$.

\bigskip

\noindent
\emph{Quantum oracle.}
Consider $t\in T_c$, and note that $Q_{k}=I-2\Pi_k$ and $Q_0=I$.
We have
\begin{align*}
H_k^{+}- H_k
& = \|Q_k\psi_k-Q_0\psi_0\|^2 - \|\psi_k-\psi_0\|^2
\\ & = \|\psi_k- 2\Pi_k\psi_k-\psi_0\|^2 - \|\psi_k-\psi_0\|^2
\\ & = 2\<\Pi_k\psi_k,\, \psi_0 - \psi_k\> + 2\<\psi_0 - \psi_k,\, \Pi_k\psi_k\> + 4\|\Pi_k\psi_k\|^2
\\ & = 4\Re\<\Pi_k\psi_k,\Pi_k\psi_0\>  
\\ & \le 4\sqrt{\beta_k}\sqrt{\gamma_k}\le 4\sqrt{\gamma_k},
\end{align*}
and hence, by the Cauchy--Schwarz inequality, we get
$H^{+}- H \le 4\sum\nolimits_k \sqrt{\gamma_k}/n \le 4/\sqrt{n}$.
\end{proof}


\subsection{Limitations of Progress Function $H$}

While the progress function $H_k^{(t)}$ is capable of providing a proof that any algorithm that fails with a small probability needs to perform at least a certain number of queries, it is not capable of providing a tight upper bound on the success probability when the algorithm is limited to a small number of queries. To see that, consider a single step of Grover's algorithm, which starts in the state $\psi^0:=\psi\otimes e_-$, where $\psi:=\sum_k\pmb{k}/\sqrt{n}$, and applies the quantum oracle $\ccQ$ to this state. We have
\[
H_k^{(1)}=\|\psi_k^1-\psi_0^1\|^2 =\|Q_k\psi^0-\psi^0\|^2 = 4/n
\]
for all $k$, and thus $H^{(1)}=4/n$.
However, the same progress function $4/n$ would be achieved if we had $\psi_0^1=\psi$ and $\psi_k^1=(1-2/\sqrt{n})\psi$. In such a case, the state of the algorithm run on $x^{(k)}$ could be
\[
\rho_k:= \psi_k^1(\psi_k^1)^* + (1-\|\psi_k^1\|^2)\pmb{k}\pmb{k}^*,
\]
and the projective measurement $\{\Pi_0,\Pi_0^\perp\}$ with $\Pi_0=\psi\psi^*$ and $\Pi_0^\perp:=I-\Pi_0$ would recognize that state $\rho_k$ is not $\psi_0^1(\psi_0^1)^*$ with probability
\[
\tr[\Pi_0^\perp\rho_k] 
= (1-\|\psi_k^1\|^2) \|\Pi_0^\perp\pmb{k}\|^2
= (4/\sqrt{n}-4/n)(1-1/n)
= \Omega(1/\sqrt{n}).
\]

\section{Case of Subconstant Success Probability}
\label{sec:proofAB}

In this section, we prove Theorem~\ref{thm:mainTechB}. Since the progress measure $H$ cannot prove the desired bound, in Section~\ref{ssec:defABprog} we define two new progress measures, $A$ and $B$, and show how the success probability of the algorithm can be bounded through them.  Then, in the brief Section~\ref{ssec:decProofEnd}, we bound $H$ through $A$, and conclude the proof of Theorem~\ref{thm:mainTechA}. In Section~\ref{ssec:chngPerCall}, we bound by how much calls to the two types of oracles can change $A$ and $B$ depending on the current value of $B$. There is a certain way for the pseudo-classical oracle to increase progress in $A$ by having regress in $B$. We handle this tradeoff in Section~\ref{ssec:finalTradeoffs}, concluding the proof of Theorem~\ref{thm:mainTechB}.

\subsection{Two Simultaneous Progress Measures}
\label{ssec:defABprog}

For ease of notation, let us write $\psi_\kappa$ instead of $\psi_\kappa   ^t$ for all $\kappa\in\{0,1,\ldots,n\}$, keeping in mind that this vector is time $t$ dependent. Thereby, other quantities that we will define through $\psi_k$ will also depend on $t$. We will explicitly state this time-dependence when necessary.

For $k\in\{1,\ldots,n\}$, let $\alpha_k\in[0,\pi/2]$ be the angle between $\psi_k$ and $\psi_0$, more precisely, let $\alpha_k:=\arccos |\<\psi_0,\psi_k/\|\psi_k\|\>|$. Note that the progress function $H_k$ satisfies
\[
H_k= 1+\|\psi_k\|^2-2\Re\<\psi_0,\psi_k\> 
 \ge 1+\|\psi_k\|^2-2\|\psi_k\|\cos\alpha_k,
\]
which is an equality if $\<\psi_0,\psi_k\>\ge 0$.
It is reasonable to hope that progress measures capable of addressing the small success case can be defined via $\|\psi_k\|$ and $\alpha_k$.
Before we decide on specific measures, let us discuss the connection between the failure probability of the algorithm and quantities $\|\psi_k\|$ and $\alpha_k$.

Let $\Pi_k$ here be the projector on answer $k$ and let $\Pi_k^\perp:=I-\Pi_k$. For Search, the set of answers is $\{1,\ldots,n\}$, therefore $\sum_k\Pi_k=I$. The failure probability of the algorithm on input $x^{(k)}$ is at least $\|\Pi_k^\perp\psi_k^\tau\|^2$.%
\footnote{It is reasonable to assume that the failure probability is exactly $\|\Pi_k^\perp\psi_k^\tau\|^2$, because, if the pseudo-classical oracle ever returns $\pmb{1}$ in the target register, the algorithm has found the solution. More formally, in such a scenario, the index register must be in state $\pmb{k}$, and the algorithm can copy this state to a designated location in the workspace register for $k$ to be returned at the end of the computation.}
Define $\theta_k:=\arcsin\|\Pi_k\psi_0\|$, therefore $\|\Pi_k^\perp\psi_0\|=\cos\theta_k$, and also note that $\sum_k(\sin\theta_k)^2=1$.

\begin{clm}
\label{clm:kFailure}
We have $\|\Pi_k^\perp\psi_k\|\ge\max\{0,\|\psi_k\|\cos(\alpha_k+\theta_k)\}$.
\end{clm}

\begin{proof}
The claim clearly holds when $\cos(\alpha_k+\theta_k)\le 0$, so let us assume $\cos(\alpha_k+\theta_k)>0$, that is,
$\cos\alpha_k\cos\theta_k > \sin\alpha_k\sin\theta_k$.
Consider unit vectors 
\[
\psi_{0,k}:=\Pi_k\psi_0/\sin\theta_k,
\qquad
\psi_{0,k}^\perp:=\Pi_k^\perp\psi_0/\cos\theta_k,
\qqAnd
 \bar\psi_0^k:=\cos\theta_k\psi_{0,k} -\sin\theta_k\psi_{0,k}^\perp,
 \]
and note that $\bar\psi_0^k$ is orthogonal to $\psi_0$.
Without loss of generality, let $\<\psi_k,\psi_0\>\ge 0$, so $\|\psi_k\|\cos\alpha_k=\<\psi_k,\psi_0\>$ and we can express $\psi_k$ as
 \[
 \psi_k = \|\psi_k\|\cos\alpha_k\psi_0 + \omega\sqrt{(\|\psi_k\|\sin\alpha_k)^2-\|\zeta\|^2}\bar\psi_0^k + \zeta
 \]
 where $\omega$ is a complex number of unit length and $\zeta$ is a vector of norm at most $\|\psi_k\|\sin\alpha_k$ and orthogonal to both $\psi_0$ and $\bar\psi^k_0$, and, thus, their linear combination $\psi_{0,k}^\perp$.
 We have
  \[
 \Pi_k^\perp\psi_k = \big(\|\psi_k\|\cos\alpha_k\cos\theta_k - \omega\sqrt{(\|\psi_k\|\sin\alpha_k)^2-\|\zeta\|^2}\sin\theta_k\big)\psi_{0,k}^\perp + \Pi_k^\perp\zeta,
 \]
and observe that
$ \<\psi_{0,k}^\perp,\Pi_k^\perp\zeta\> = \<\Pi_k^\perp\psi_{0,k}^\perp,\zeta\>
=\<\psi_{0,k}^\perp,\zeta\> = 0$.
Since $\cos\alpha_k\cos\theta_k > \sin\alpha_k\sin\theta_k$, the quantity  $\|\Pi_k^\perp\psi_k\|$ is minimized by choosing phase $\omega=1$ and vector $\zeta=0$.
\end{proof}

Hence, the average failure probability is at least
\[
\min_\Theta \sum_k(\|\psi_k\|\max\{0,\cos(\alpha_k+\theta_k)\})^2/n
\]
where the minimization is taken over all tuples $\Theta=(\theta_1,\ldots,\theta_k)$ such that $\sum_k(\sin\theta_k)^2=1$.
One might consider this quantity as the progress function---and one might indeed have to when trying to show exact bounds for hybrid pseudo-classical and quantum query algorithms---but this quantity seems to be difficult to handle. 
Instead, let us introduce progress measures
\[
A_k :=  
 (\|\psi_k\|\cos\alpha_k)^2,
\quad
B_k :=
 (\|\psi_k\|\sin\alpha_k)^2,
 \quad
 A:=\sum\nolimits_kA_k/n,
 \qAnd
 B:=\sum\nolimits_kB_k/n,
\]
which are all functions of time $t$.

When proving lower bounds, we assume that $k\in\{1,\ldots,n\}$ is chosen uniformly at random, and the input is $x^{(k)}$. Any lower bounds proven in this average-case scenario also hold for the worst-case.

\begin{lem}
\label{lem:avgFail}
The average-case failure probability of the algorithm is at least $A^{(\tau)}-\frac{1}{n}-\frac{2\sqrt{B^{(\tau)}}}{\sqrt{n}}$.
\end{lem}

\begin{proof}
Let $\gamma_k:=(\sin\theta_k)^2$, so, from Claim~\ref{clm:kFailure}, we get that
\[
\|\Pi_k^\perp\psi_k\|\ge\max\big\{0,\sqrt{(1-\gamma_k)A_k}-\sqrt{\gamma_kB_k}\big\}.
\]
Regardless whether $(1-\gamma_k)A_k\ge \gamma_kB_k$ or not, we have
\[
\|\Pi_k^\perp\psi_k\|^2 
\ge \sqrt{(1-\gamma_k)A_k}\big(\sqrt{(1-\gamma_k)A_k}-2\sqrt{\gamma_kB_k}\big)
\ge A_k-\gamma_k-2\sqrt{\gamma_kB_k}.
\]
Hence, by the Cauchy--Schwarz inequality, we have
\[
\frac{1}{n}\sum\nolimits_k\|\Pi_k^\perp\psi_k\|^2
\ge \frac{1}{n}\sum\nolimits_k A_k-\frac{1}{n}\sum\nolimits_k\gamma_k-\frac{2}{n}\sum\nolimits_k\sqrt{\gamma_kB_k}  
\ge A-\frac{1}{n}-\frac{2\sqrt{B}}{\sqrt{n}}.
\qedhere
\]
\end{proof}

Initially, we have $A^{(0)}=1$ and $B^{(0)}=0$. In order for algorithm to succeed, it has to reduce the quantity $A$. As we are about to discuss, increasing the quantity $B$ is also beneficial.

Queries to the quantum oracle do not change $\|\psi_k\|$ and thus leave $A+B$ unchanged. We will show in Section~\ref{ssec:chngPerCall} that quantum queries can reduce $A$ by at most about $4\sqrt{B/n}$ while increasing $B$ by the same quantity.
The situation is more complicated for pseudo-classical queries. It is clear that we can keep both $A$ and $B$ essentially unchanged by swapping useful memory with some ``junk'' memory and then feeding that junk memory to the oracle. We will show that, essentially, while pseudo-classical queries can also reduce $A$ by about $2\sqrt{B/n}$, the maximum reduction in $A$ also reduces $B$ back to close to $0$.

Thus, we may think of $B$ as a sort of potential. The quantum oracle can simultaneously increase this potential $B$ by order of $\sqrt{B/n}$ and decrease $A$ by the same amount. A query to the pseudo-classical oracle can increase $B$ by at most $1/n$, but a part $z$ of this potential $B$ can be also spent to decrease $A$ by about $2\sqrt{z/n}$.

\subsection{Conclusion for the Proof of Theorem~\ref{thm:mainTechA}}
\label{ssec:decProofEnd}

In Claim~\ref{clm:progQ} we have already bounded by how much a single query can change the progress measure $H$. Here we relate its final value to the failure probability $\epsilon$ in the finding version of Search. We omit the time superscript $\tau$ for $H,A,B,\psi$.

Observe that we have $\sqrt{H_k}=\|\psi_k-\psi_0\|\ge 1-|\<\psi_k,\psi_0\>|=1-\sqrt{A_k}$, which implies that $A_k \ge (1-\sqrt{H_k})^2$ whenever $H_k\le 1$. If we had $H_k\le 1$ for all $k$, this would in turn imply 
\[
A
\ge
 \frac{1}{n}\sum\nolimits_k (1-\sqrt{H_k})^2
 = 1 + H - \frac{2}{n}\sum\nolimits_k \sqrt{H_k} 
 \ge 1 - 2\sqrt{H} + H
 = (1 - \sqrt{H})^2,
\]
and, together with $B\le 1$, we would get the desired result by using Lemma~\ref{lem:avgFail}. However, it may happen that $H_k > 1$ for some $k$, causing additional technical hurdles.
On the bright side, intuitively, having $H_k > 1$ should not impede the proof, because, in essence, what we are trying to show is that, in order for the error to be small, the average value of $H_k$ must be at least close to $1$. We overcome these hurdles by, instead of using Lemma~\ref{lem:avgFail} directly, closely following its proof.

As in the proof of Lemma~\ref{lem:avgFail}, let $\Pi_k$ be the projector on answer $k$, let $\Pi_k^\perp:=I-\Pi_k$, and let $\gamma_k:=\|\Pi_k\psi_0\|^2$.
Hence, just as in the proof of Lemma~\ref{lem:avgFail}, we have
\[
\|\Pi_k^\perp\psi_k\|^2 
\ge A_k-\gamma_k-2\sqrt{\gamma_kB_k}
\ge A_k-3\sqrt{\gamma_k},
\]
where, for the latter inequality, we have used $B_k\le 1$.
Here, for $k$ such that $H_k\le 1$, we can bound $\|\Pi_k^\perp\psi_k\|^2$ further by using $A_k\ge (1-\sqrt{H_k})^2$.

Let 
$K' := \{k\colon H_k\le 1\} \subseteq\{1,\ldots,n\}$ and $n':=|K'|$, and let us define $H':=\sum_{k\in K'}H_k/n$, which satisfies $ H' \le n'/n$.
By the Cauchy--Schwarz inequality, we have
\begin{align*}
\epsilon
\,& =\frac{1}{n}\sum_k\|\Pi_k^\perp\psi_k\|^2
\\ &
\ge \frac{1}{n}\sum_{k\in K'}
\Big((1-\sqrt{H_k})^2-3\sqrt{\gamma_k}\Big)
\\ & =
\frac{n'}{n}+H'
 - \frac2n\sum_{k\in K'}\sqrt{H_k}-\frac3n\sum_{k\in K'}\sqrt{\gamma_k}
 \\ & \ge 
 \frac{n'}{n}+H'
 - \frac{2\sqrt{n'}}{n}\sqrt{nH'}-\frac{3\sqrt{n'}}{n}
  \\ & \ge 
\Big(\sqrt{\frac{n'}{n}}-\sqrt{H'}\Big)^2
 - \frac{3}{\sqrt{n}}.
\end{align*}
Hence, $H'$ must satisfy
\[
\sqrt{\frac{n'}{n}}-\sqrt{H'} \le \sqrt{\epsilon + 3/n^{1/2}} \le \sqrt{\epsilon} + 2/n^{1/4},
\]
and thus
\[
H' \ge \frac{n'}{n} - 2\sqrt{\epsilon} - 4/ n^{1/4}.
\]
Finally, since $H_k>1$ for $k\notin K'$, we have
\[
H = \frac{1}{n}\sum_{k\in \{1,\ldots,n\}\setminus K'}H_k + \frac{1}{n}\sum_{k\in K'}H_k
\ge \frac{n-n'}{n} + H'
\ge 1 - 2\sqrt{\epsilon} - 4/ n^{1/4}.
\]

\subsection{Change in the Progress Measures under Oracle Calls}
\label{ssec:chngPerCall}

Here we bound possible changes under a single oracle call, and the section is devoted to establishing the following claim.

\begin{clm}
\label{clm:ABbound}
If $t$-th oracle call is pseudo-classical, i.e., $t\in T_c$, there exists $\wo{t} \in [0, B^{(t-1)}]$ such that 
\[
A^{(t)}\ge A^{(t-1)}-2/n-2\sqrt{\wo{t}/n}
\qqAnd
B^{(t)}\le B^{(t-1)}-\wo{t}+1/n.
\]
If $t$-th oracle call is quantum, i.e., $t\in T_q$, we have
\[
A^{(t)}\ge A^{(t-1)}-4/n-4\sqrt{B^{(t-1)}/n}
\qqAnd
B^{(t)}\le B^{(t-1)}+4/n+4\sqrt{B^{(t-1)}/n}.
\]
\end{clm}

\begin{proof}
Before we address the two oracles individually, let us introduce some notation shared between the two scenarios. The proof has many resemblances to the proof of Claim~\ref{clm:kFailure}.
 For sake of brevity, let $A,B,A_k,B_k,\psi_k,\psi_0$ denote the corresponding quantities at time $t-1$, while $A^+,B^+,A_k^+,B_k^+$ denote them at time $t$ (see Footnote~\ref{fn:superPlus}).

For $k\in\{1,\ldots,n\}$, define the projector $\Pi_k := \pmb{k}\pmb{k}^*$ if $t\in T_c$ and $\Pi_k:=\pmb{k}\pmb{k}^*\otimes e_-e_-^*$ if $t\in T_q$, and let $\Pi_k^\perp:=I-\Pi_k$. Define 
$\gamma_k:=\|\Pi_k\psi_0\|^2$, and note that $\sum_k\gamma_k\le 1$ (with equality when $t\in T_c$). Let 
\[
\psi_{0,k}:=\Pi_k\psi_0/\sqrt{\gamma_k},
\qquad
\psi_{0,k}^\perp:=\Pi_k^\perp\psi_0/\sqrt{1-\gamma_k},
\qqAnd
 \bar\psi_0^k:=\sqrt{1-\gamma_k}\psi_{0,k} -\sqrt{\gamma_k}\psi_{0,k}^\perp,
\]
Without loss of generality,%
\footnote{Because neither $A_k$ nor $B_k$ depends on the global phase of $\psi_k$.}
 let $\<\psi_k,\psi_0\>\ge 0$, so $\sqrt{A_k}=\<\psi_k,\psi_0\>$ and we can express $\psi_k$ as
 \[
 \psi_k = \sqrt{A_k} \psi_0 + \omega\sqrt{B_k-\|\zeta\|^2}\bar\psi_0^k + \zeta
 \]
 where $\omega$ is a complex number of unit length and $\zeta$ is a $k$-dependent vector of norm at most $\sqrt{B_k}$ and orthogonal to both $\psi_{0}$ and $\bar\psi_{0}^{k}$. 
 Let $\zeta_k^\perp:=\Pi_k^\perp\zeta$, which is orthogonal to $\psi_{0,k}^\perp$.
 
\bigskip

\noindent
\emph{Pseudo-classical oracle.}
Let $t\in T_c$, and note that $P_{0,k}=\Pi_k^\perp\otimes\pmb{0}$. 
We have
\[
 P_{k,0}\psi_k = \Big(\sqrt{(1-\gamma_k)A_k} - \omega\sqrt{\gamma_k(B_k-\|\zeta\|^2)}\Big)\psi_{0,k}^\perp\otimes\pmb{0} + \zeta_k^\perp\otimes\pmb{0}.
\]
and $ P_{0,0}\psi_0=\psi_0\otimes\pmb{0}$.
For progress measure $A_k^+$, we have
\begin{align*}
\sqrt{A_k^+}=|\< P_{0,0}\psi_0, P_{k,0}\psi_k\>|
& = \sqrt{1-\gamma_k} 
\Big|\sqrt{(1-\gamma_k)A_k} - \omega\sqrt{\gamma_k(B_k-\|\zeta\|^2)}\Big| \\
 & \ge \sqrt{1-\gamma_k} \max\{0,\sqrt{(1-\gamma_k)A_k}-\sqrt{\gamma_k(B_k-\|\zeta_k^\perp\|^2)}\},
\end{align*}
from which, using $A_k\le 1$ and $1-\gamma_k\le 1$, we get
\[
A_k^+
\ge (1-\gamma_k)^{3/2} \sqrt{A_k}\big(\sqrt{(1-\gamma_k)A_k}-2\sqrt{\gamma_k(B_k-\|\zeta_k^\perp\|^2)}\big)
\ge A_k-2\gamma_k - 2\sqrt{\gamma_k(B_k-\|\zeta_k^\perp\|^2)}.
\]
For progress measure $B_k^+$, we have
\begin{align*}
B_k^+  = \|P_{k,0}\psi_k\|^2-A_k^+ 
& = 
 \big|\sqrt{(1-\gamma_k)A_k} - \omega\sqrt{\gamma_k(B_k-\|\zeta\|^2)}\big|^2 
 + \|\zeta_k^\perp\|^2 \\
 & \qquad
 - (1-\gamma_k)
\big|\sqrt{(1-\gamma_k)A_k} - \omega\sqrt{\gamma_k(B_k-\|\zeta\|^2)}\big|^2 \\
 & = 
\gamma_k \big|\sqrt{(1-\gamma_k)A_k} - \omega\sqrt{\gamma_k(B_k-\|\zeta\|^2)}\big|^2 
 + \|\zeta_k^\perp\|^2 \\
  & \le  \gamma_k + \|\zeta_k^\perp\|^2,
\end{align*}
where we have used that $A_k+B_k\le 1$.

Recall that $B=\sum_k B_k/n$, and define $\wo{t}:=B-\sum_k \|\zeta_k^\perp\|^2/n \ge 0$.  
By the Cauchy--Schwarz inequality, we get the claimed $A^+ \ge A - 2/n - 2 \sqrt{\wo{t}/n}$ and $B^+ \le B-\wo{t}+1/n$.

\bigskip

\noindent
\emph{Quantum oracle.}
Let $t\in T_q$, and note that $Q_{k}= I-2\Pi_k$.
We have
  \begin{align*}
\sqrt{A_k^+} = |\<\psi_0,Q_k\psi_k\>| 
 & = \big|\sqrt{A_k} - 2  \<\Pi_k\psi_0,\Pi_k\psi_k\>\big| \\
 & = \big|\sqrt{A_k} - 2  \sqrt{\gamma_k} 
 \big(\sqrt{\gamma_kA_k} + \omega\sqrt{(1-\gamma_k)(B_k-\|\zeta_k\|^2)}\big)\big| \\
 & \ge \max\big\{0,
(1-2\gamma_k)\sqrt{A_k} - 2  \sqrt{\gamma_kB_k}
 \big\},
 \end{align*} 
and therefore
  \[
A_k^+
 \ge (1-2\gamma_k)\sqrt{A_k}\big((1-2\gamma_k)\sqrt{A_k} - 4  \sqrt{\gamma_kB_k}\big)
 \ge A_k-4\gamma_k - 4  \sqrt{\gamma_kB_k}.
 \]
 By the Cauchy--Schwarz inequality, we get the claimed $A^+ \ge A - 4/n - 4 \sqrt{B/n}$. 
Since $A^++B^+=A+B$ for calls to the quantum oracle, this concludes the proof.
\end{proof}

 \subsection{Tradeoffs in the Evolution of Progress Measures}
 \label{ssec:finalTradeoffs}

The rest of the proof becomes simpler if we effectively treat inequalities of Claim~\ref{clm:ABbound} as equalities. 
Formally, let $\wo{t}$ be as in Claim~\ref{clm:ABbound}, let $a_0:=A^{(0)}=1$, $b_0:=B^{(0)}=0$, and, for all $t\in\{1,...,\tau\}$, define values $a_t$ and $b_t$ recursively as 
\[
a_t := a_{t-1} -
\begin{cases}
2/n+2\sqrt{\wo{t}/n} & \text{if }t\in T_c, \\
4/n+4\sqrt{b_{t-1}/n} & \text{if }t\in T_q,
\end{cases}
\qquad
b_t := b_{t-1} + 
\begin{cases}
1/n-\wo{t} & \text{if }t\in T_c, \\
4/n+4\sqrt{b_{t-1}/n} & \text{if }t\in T_q.
\end{cases}
\]
Note that $A^{(t)}\ge a_t$ and $B^{(t)}\le b_t$ for all $t$, and thus also $\wo{t}\le  B^{(t-1)} \le b_{t-1}$.

We want to show that the final values $a_\tau$ and $b_\tau$ cannot be far away from their initial values, $1$ and $0$, respectively. Let us start by a bound that essentially limits by how much quantum queries can increase $b_t$ and a bound on its maximum possible value over all queries.

\begin{lem}
\label{lem:SumSqB}
We have $\sum_{t\in T_q}\sqrt{b_{t-1}/n} \le \tau_q (\sqrt{\tau_c}+\tau_q-1)/n$ and $b_\tau\le(\sqrt{\tau_c}+2\tau_q)^2/n$.
\end{lem}

\begin{proof}
For $i\in\{1,\ldots,\tau_q\}$, let $t_i$ be the $i$-th smallest ``clock'' value in $T_q$. The largest possible value of $b_{t_1}$ is achieved when the first $\tau_c$ queries are all pseudo-classical and $\wo{t}=0$ for all of them. That is, $t_1=\tau_c+1$ and $b_{t_1-1}=\tau_c/n$. By the same argument, for every $i$, the largest possible value of $b_{t_i}$ is achieved when all pseudo-classical queries are performed before any quantum queries. Suppose that it is the case, and $b_{\tau_c}=\tau_c/n$. The remaining queries are all quantum, and, by induction, it is easy to see that $b_{\tau_c+i}=(\sqrt{\tau_c}+2i)^2/n$. Indeed, as the base case, it clearly holds for $i=0$. As for the inductive step, we have
\begin{align*}
b_{\tau_c+i} & = b_{\tau_c+i-1} + 4/n+4\sqrt{b_{\tau_c+i-1}/n} \\
 & = \big((\sqrt{\tau_c}+2i-2)^2 + 4+4(\sqrt{\tau_c}+2i-2)\big)/n \\
& =(\sqrt{\tau_c}+2i)^2/n.
\end{align*}
Hence, we have
\[
\sum_{t\in T_q}\sqrt{b_{t-1}/n}
\le \sum_{i=0}^{\tau_q-1} (\sqrt{\tau_c}+2i)/n
= \tau_q (\sqrt{\tau_c}+\tau_q-1)/n.
\]

As for maximizing the value of $b_\tau$, note that each pseudo-classical query can increase the value of $b_t$ by at most $1/n$, while, for quantum queries, the increase grows along the value of $b_t$ itself. Hence, again, for maximizing $b_\tau$, we can assume that all pseudo-classical queries precede any quantum queries, and the analysis above applies.
\end{proof}

The above lemma effectively tells us that all quantum queries in total cannot reduce the progress measure $a_t$ by more than $\mathrm{O}((\tau_q^2 + \tau_q\sqrt{\tau_c})/n)$. But it also has an additional purpose. Large values of $b_t$ can contribute to large reductions of $a_t$ under pseudo-classical queries. However, such reductions ``spend'' potential $b_t$, and Lemma~\ref{lem:SumSqB} limits by how much quantum queries can ``charge'' this potential over the whole algorithm. Thus we can show that all pseudo-classical queries in total cannot reduce the progress measure $a_t$ by more than $\mathrm{O}((\tau_c + \tau_q\sqrt{\tau_c})/n)$.

\begin{clm}
We have $\sum_{t\in T_c}\sqrt{\wo{t}/n} \le (\tau_c+2\sqrt{\tau_c}\tau_q)/n$.
\end{clm}

\begin{proof}
Let
\[
\Delta_t := b_t - b_{t-1} = 
\begin{cases}
1/n-\wo{t} & \text{if }t\in T_c, \\
4/n+4\sqrt{b_{t-1}/n} & \text{if }t\in T_q.
\end{cases}
\]
Since all the values $b_t$ are non-negative, we have
\begin{align*}
\sum_{t\in T_c\colon\Delta_t<0}|\Delta_t| 
& \le \sum_{t\colon\Delta_t> 0}\Delta_t \\
& = \sum_{t\in T_c\colon\Delta_t>0}\Delta_t
+ \sum_{t\in T_q} (4/n+4\sqrt{b_{t-1}/n}) \\
& \le \sum_{t\in T_c\colon\Delta_t>0}\Delta_t + 4\tau_q (\sqrt{\tau_c}+\tau_q)/n,
\end{align*}
where the last inequality is due to Lemma~\ref{lem:SumSqB}.
Hence, we have
\begin{align*}
\sum_{t\in T_c}\wo{t} 
& = \frac{\tau_c}{n} - \sum_{t\in T_c}\Delta_t \\
& = \frac{\tau_c}{n} + \sum_{t\in T_c\colon\Delta_t<0}|\Delta_t| - \sum_{t\in T_c\colon\Delta_t > 0}\Delta_t \\
& \le \frac{\tau_c+4\tau_q (\sqrt{\tau_c}+\tau_q)}{n} \\
& = (\sqrt{\tau_c}+2\tau_q)^2/n .
\end{align*}
By the Cauchy--Schwarz inequality,
\[
\sum_{t\in T_c}\sqrt{\wo{t}/n}\le \sqrt{\tau_c/n} \sqrt{\sum_{t\in T_c} \wo{t}}
\le (\tau_c+2\sqrt{\tau_c}\tau_q)/n.
\qedhere
\]
\end{proof}

Now we are ready to bound the success probability.
As a part of Lemma~\ref{lem:SumSqB}, we have already shown that $\sqrt{b_\tau}\le(\sqrt{\tau_c}+2\tau_q)/\sqrt{n}$, and now we can also bound $a_\tau$ as
\begin{align*}
a_\tau 
& = 1 - \sum_{t\in T_c}(2/n+2\sqrt{\wo{t}/n}) - \sum_{t\in T_q}(4/n+4\sqrt{b_{t-1}/n}) \\
& \ge 1 - 2\tau_c /n- 2(\tau_c+2\sqrt{\tau_c}\tau_q)/n - 4\tau_q/n-4\tau_q (\sqrt{\tau_c}+\tau_q-1)/n \\
& = 1 - 4(\sqrt{\tau_c}+\tau_q)^2/n.
\end{align*}
By Lemma~\ref{lem:avgFail}, the success probability of the algorithm is at most 
\begin{align*}
1-a_\tau+1/n + 2\sqrt{b_\tau/n}
& \le 4(\sqrt{\tau_c}+\tau_q)^2/n + 1/n + 2(\sqrt{\tau_c}+2\tau_q)/n \\
& = [4(\sqrt{\tau_c}+\tau_q)^2 + 1 + 2\sqrt{\tau_c}+4\tau_q]/n \\
& \le (2\sqrt{\tau_c}+2\tau_q+1)^2/n.
\end{align*}

\section{Open problems}

When the task is to search for a unique marked element, Zalka showed that Grover's algorithm is exactly optimal~\cite{zalka:Grover}. 
That is, given $\tau_q$ quantum queries, an algorithm can find the marked element with probability $\sin^2((1+2\tau_q)\arcsin(1/\sqrt{n}))$, and no better than that.
Can we show a similar exact bound for hybrid quantum-classical algorithms that have $\tau_c$ classical and $\tau_q$ quantum queries? In particular, can we show that the best thing an algorithm can do is, first, to randomly choose and classically query $\tau_c$ indices, and then run Grover's algorithm on the remaining $n-\tau_c$ indices?

\section*{Acknowledgements}

I would like to thank Fran\c{c}ois Le~Gall and Aleksandrs Belovs for fruitful discussions. I would like to thank Aleksandrs Belovs for bringing to my attention the result by Ambainis et al.~on semi-classical oracles.
Finally, I would like to thank anonymous reviewers for pointing out some flaws in the earlier version of this work, in particular, the proof of Theorem~\ref{thm:mainTechA} in Section~\ref{ssec:decProofEnd}.

This work was supported by JSPS KAKENHI Grant Number JP20H05966 and MEXT Quantum Leap Flagship Program (MEXT Q-LEAP) Grant Number JPMXS0120319794.

{
\small
\newcommand{\etalchar}[1]{$^{#1}$}

}

\end{document}